\documentclass[11pt]{article}
\usepackage[margin=1in]{geometry}
\usepackage[utf8]{inputenc}
\usepackage[T1]{fontenc}
\usepackage{lmodern}
\usepackage{microtype}
\usepackage{amsmath,amssymb,amsthm}
\usepackage{booktabs,array}
\usepackage{enumitem}
\usepackage[numbers,sort&compress]{natbib}
\usepackage{xurl}
\usepackage[hidelinks]{hyperref}
\newtheorem{theorem}{Theorem}[section]
\newtheorem{lemma}[theorem]{Lemma}
\newtheorem{proposition}[theorem]{Proposition}
\newtheorem{corollary}[theorem]{Corollary}
\theoremstyle{definition}

\theoremstyle{remark}
\newtheorem{remark}[theorem]{Remark}
\newcommand{\Dk}{\mathsf{D}}
\newcommand{\Dstar}{\mathsf{D}^{*}}
\newcommand{\C}{C_5^{3}}
\newcommand{\supp}{\operatorname{supp}}
\newcommand{\rank}{\operatorname{rank}}
\newcommand{\zz}{\mathsf{z}}
\newcommand{\sle}[1]{\mathsf{s}_{\le #1}}
\ifdefined\pdfinfoomitdate\pdfinfoomitdate=1\fi
\ifdefined\pdftrailerid\pdftrailerid{}\fi
\ifdefined\pdfsuppressptexinfo\pdfsuppressptexinfo=15\fi

\hypersetup{pdftitle={The fourth generalized Davenport constant of C\_5\^{}3},pdfauthor={Sze Chun Yiu},pdfsubject={Generalized Davenport constants of the elementary abelian group of order 125},pdfkeywords={generalized Davenport constants, zero-sum sequences, elementary abelian groups, computer-assisted proof}}
\title{The fourth generalized Davenport constant of $C_5^3$}
\author{Sze Chun Yiu\\Stockholm University\\\texttt{sze-chun.yiu@fysik.su.se}}
\date{September 2026}
\begin{document}
\maketitle
\begin{abstract}
For a finite abelian group $G$ and $k\ge1$, the generalized Davenport constant $\Dk_k(G)$ is the least $\ell$ such that every sequence over $G$ of length at least $\ell$ has $k$ pairwise disjoint nonempty zero-sum subsequences. A theorem of Freeze and Schmid gives $\Dk_k(\C)\ge 5k+10$ for every $k\ge2$. We prove the matching upper bound: $\Dk_4(\C)=30$, and hence $\Dk_k(\C)=5k+10$ for every $k\ge2$, so the Freeze--Schmid bound is attained by $\C$ from $k=2$ onward, as it is by $C_2^3$ and unlike $C_3^3$. The proof is finite and computer-assisted. The remaining case reduces to showing that every zero-sum sequence of length $31$ over $\C$ contains a nonempty zero-sum subsequence of length at most five. A saturation argument confines the multiplicities of a hypothetical counterexample to $\{1,2,4\}$, its support pattern to one of $60$ solutions of two linear equations, and its geometry to one of $78$ rank/plane branches normalized to a standard basis; an exhaustive search exhausts every branch with no survivor. The search was carried out by three independently written implementations, and the branch cover was regenerated by separate programs from the lemmas alone; two further machine-verified values, $\Dk_3(\C)=25$ and $\sle{6}(\C)=24$, enter the second statement, and their records accompany the paper.
\end{abstract}



\section{Introduction}\label{sec:intro}

Let $G$ be a finite abelian group, written additively. A \emph{sequence} over $G$ is a finite multiset of elements of $G$; we write it multiplicatively as $S=g_1\cdots g_\ell$, call $|S|=\ell$ its length, $\sigma(S)=g_1+\dots+g_\ell$ its sum, and $\supp(S)$ the set of elements that occur in it. A \emph{subsequence} is a sub-multiset, and $S$ is \emph{zero-sum} if $\sigma(S)=0$. For $k\ge1$ the \emph{$k$-th generalized Davenport constant} $\Dk_k(G)$, introduced by Halter-Koch \cite{HalterKoch1992}, is the least $\ell$ such that every sequence over $G$ of length at least $\ell$ has $k$ pairwise disjoint nonempty zero-sum subsequences; $\Dk_1(G)=\Dk(G)$ is the classical Davenport constant. For $\ell\ge1$ let $\sle{\ell}(G)$ be the least $m$ such that every sequence of length at least $m$ has a nonempty zero-sum subsequence of length at most $\ell$, and let $\eta(G)=\sle{\exp(G)}(G)$. Background on these invariants is in the survey of Gao and Geroldinger \cite{GaoGeroldinger2006} and the monograph of Geroldinger and Halter-Koch \cite{GeroldingerHalterKoch2006}.

Exact values of $\Dk_k(G)$ for all $k$ are known for groups of rank at most two, where $\Dk_k(C_{n_1}\oplus C_{n_2})=n_1+kn_2-1$ for $n_1\mid n_2$, and for $p$-groups with $\Dk(G)\le2\exp(G)-1$, where $\Dk_k(G)=\Dk(G)+(k-1)\exp(G)$ (see \cite[Section~6.1]{GeroldingerHalterKoch2006} and \cite{DelormeOrdazQuiroz2001}); among elementary $p$-groups of rank at least three they are known for $C_2^r$ with $r\le5$ \cite{DelormeOrdazQuiroz2001,FreezeSchmid2010} and for $C_3^3$ \cite{BhowmikSchlagePuchta2007}. Freeze and Schmid \cite{FreezeSchmid2010} proved that $k\mapsto\Dk_k(G)$ is eventually an arithmetic progression with difference $\exp(G)$ and gave lower bounds that are sharp in many cases. Their Theorem~4.1 states that for $G=C_{n_1}\oplus\dots\oplus C_{n_r}$ with $1<n_1\mid\dots\mid n_r$, an integer $s\ge2$ and $t\in[1,r]$ with $s(s-1)/2\le r-t+1$,
\begin{equation}\label{eq:FS}
\Dk_k(G)\ \ge\ \Dstar(G)+s\Bigl\lfloor\frac{n_t}{2}\Bigr\rfloor+\delta+(k-2)n_r\qquad(k\ge2),
\end{equation}
where $\Dstar(G)=1+\sum_{i}(n_i-1)$ and $\delta=1$ if $n_t$ is odd and $\delta=0$ otherwise. For $G=\C$, $s=3$ and $t=1$ the hypothesis $3\le3$ holds, $\Dstar(\C)=13$, $\lfloor 5/2\rfloor=2$ and $\delta=1$, so
\begin{equation}\label{eq:lower}
\Dk_k(\C)\ \ge\ 13+6+1+5(k-2)\ =\ 5k+10\qquad(k\ge2).
\end{equation}
The question left open by \eqref{eq:lower} is whether this bound is attained. We show that it is, for every $k\ge2$.

\begin{theorem}\label{thm:main}
\leavevmode
\begin{enumerate}[label=\textup{(\alph*)},leftmargin=2em]
\item Every zero-sum sequence over $\C$ of length $31$ has a nonempty zero-sum subsequence of length at most five.
\item $\Dk_4(\C)=30$, and consequently $\Dk_k(\C)=5k+10$ for every $k\ge2$.
\end{enumerate}
\end{theorem}

Part~(a) is proved by a finite, computer-assisted case analysis. Part~(b) follows from part~(a), the lower bound \eqref{eq:lower}, the value $\Dk_2(\C)=20$ (Proposition~\ref{prop:D2}), and two further values, $\Dk_3(\C)=25$ and $\sle{6}(\C)=24$, which were established by exhaustive machine computation and whose records accompany this paper (Section~\ref{subsec:machine}). Both parts use results from the literature as premises, each named where it is used: the Freeze--Schmid bound and recurrence \cite{FreezeSchmid2010}, Olson's value $\Dk(\C)=13$ \cite{Olson1969}, the values $\eta(C_5^2)=13$ and $\eta(\C)=33$ together with the structure of the extremal sequences for $\eta(\C)$ \cite{GaoHouSchmidThangadurai2007,FanGaoWangZhongZhuang2012}, and, for one of two routes to $\Dk_2(\C)=20$, a lemma from a preprint of Zhao \cite{Zhao2025}. No new general theory is claimed; the contribution is the exact value for this group, the finite reduction that makes it decidable, and the verification path.

The value places $\C$ on the side of the two smallest primes' rank-three groups that attain the bound: for $C_2^3$ one has $\Dk_k(C_2^3)=2k+3$ for all $k\ge2$ (Remark~\ref{rem:C23}), which is \eqref{eq:FS} with $s=3$, $t=1$, $\delta=0$; for $C_3^3$ one has $\Dk_1=7$, $\Dk_2=11$ and $\Dk_k=3k+6$ for $k\ge3$ \cite{BhowmikSchlagePuchta2007}, as recorded in \cite[Remark~5.3]{FreezeSchmid2010}, while \eqref{eq:FS} gives $3k+5$, attained at $k=2$ and exceeded from $k=3$ onward. So among $C_p^3$ with $p\in\{2,3,5\}$ the Freeze--Schmid bound is attained for all $k\ge2$ exactly when $p\ne3$. We have no explanation of this and record it as data.

The paper is organized as follows. Section~\ref{sec:ingredients} fixes notation, records the ingredients, and reduces part~(b) to part~(a). Section~\ref{sec:obstruction} derives the structure of a hypothetical counterexample to part~(a): saturation, the multiplicity grammar of $60$ patterns, a projective-line restriction, and the exclusion of small supports. Section~\ref{sec:cover} partitions the remaining cases into $78$ branches, each normalized to a standard basis. Section~\ref{sec:search} describes the exhaustive search, the three implementations, and the verification layers. Section~\ref{sec:proof} assembles the proof, Section~\ref{sec:prior} positions the result, and Section~\ref{sec:deferred} records the provenance of the computations and the limits of the claim.

\section{Notation and ingredients}\label{sec:ingredients}

Throughout, $G=\C=(\mathbb{Z}/5\mathbb{Z})^3$ unless stated otherwise, so $\exp(G)=5$ and $G$ is a vector space over the field with five elements. For a sequence $S$ over $G$ and $g\in G$ let $v_g(S)$ be the multiplicity of $g$ in $S$; the \emph{support size} of $S$ is $|\supp(S)|$, and $\rank(S)$ is the dimension of the span of $\supp(S)$. For a sequence $S$ let $\zz(S)$ denote the maximum number of pairwise disjoint nonempty zero-sum subsequences of $S$, so that $\Dk_k(G)$ is the least $\ell$ with $\zz(S)\ge k$ for all $|S|\ge\ell$. A sequence is \emph{short-zero-free} if it has no nonempty zero-sum subsequence of length at most five; the word ``short'' always refers to length at most $\exp(G)=5$. When elements of $G$ are listed in a record we encode $(x,y,z)\in(\mathbb{Z}/5\mathbb{Z})^3$ as the integer $25x+5y+z$, so that $e_1=25$, $e_2=5$, $e_3=1$.

\subsection{Two elementary lemmas}

\begin{lemma}\label{lem:Dk-char}
For every finite abelian group $G$ and $k\ge1$,
\[
\Dk_k(G)=\max\{\,|B| : B \text{ a zero-sum sequence over } G \text{ with } \zz(B)\le k\,\}.
\]
\end{lemma}

This characterization is standard (compare \cite[Section~6.1]{GeroldingerHalterKoch2006} and \cite{FreezeSchmid2010}); we include the short argument because we use both directions.

\begin{proof}
Let $N$ be the maximum length of a sequence $S$ with $\zz(S)\le k-1$, so that $\Dk_k(G)=N+1$, and let $M$ be the maximum on the right-hand side; both are finite because a sequence of length at least $\Dk_{k+1}(G)$ has $k+1$ disjoint zero-sum subsequences. If $|S|=N$ and $\zz(S)\le k-1$, then $B=S\cdot(-\sigma(S))$ is zero-sum, and among pairwise disjoint zero-sum subsequences of $B$ at most one contains the appended element while the others are disjoint zero-sum subsequences of $S$; hence $\zz(B)\le k$ and $M\ge N+1$. Conversely, let $B$ be zero-sum with $\zz(B)\le k$ and $|B|=M$, and let $S$ be $B$ with one element $g$ removed. If $S$ had $k$ pairwise disjoint nonempty zero-sum subsequences, their complement in $B$ would be zero-sum and nonempty (it contains $g$), giving $\zz(B)\ge k+1$. So $\zz(S)\le k-1$ and $N\ge M-1$.
\end{proof}

\begin{lemma}[{\cite[Proposition~3.1]{FreezeSchmid2010}}]\label{lem:recurrence}
For every finite abelian group $G$, $k\ge1$ and $\ell\ge1$,
\[
\Dk_{k+1}(G)\ \le\ \max\{\,\Dk_k(G)+\ell,\ \sle{\ell}(G)-1\,\}.
\]
\end{lemma}

\begin{proof}
Let $S$ have length $n\ge\max\{\Dk_k(G)+\ell,\sle{\ell}(G)-1\}$ and put $B=S\cdot g$ with $g=-\sigma(S)$. Since $|B|\ge\sle{\ell}(G)$, $B$ has a nonempty zero-sum subsequence $A$ of length at most $\ell$. If $g\notin A$, then $A$ is a zero-sum subsequence of $S$ and $S\cdot A^{-1}$ has length at least $n-\ell\ge\Dk_k(G)$, so it contains $k$ disjoint nonempty zero-sum subsequences, and with $A$ we get $\zz(S)\ge k+1$. If $g\in A$, let $A'=A\cdot g^{-1}$; then $A'$ is a subsequence of $S$ with $\sigma(A')=-g=\sigma(S)$ and $|A'|\le\ell-1$, so $S\cdot A'^{-1}$ is zero-sum of length at least $n-\ell+1\ge\Dk_k(G)+1$. Removing one element leaves at least $\Dk_k(G)$ terms, hence $k$ disjoint nonempty zero-sum subsequences, whose complement in $S\cdot A'^{-1}$ contains the removed element and is a nonempty zero-sum sequence; again $\zz(S)\ge k+1$.
\end{proof}

\subsection{The value $\Dk_2(\C)=20$}

The lower bound $\Dk_2(\C)\ge20$ is \eqref{eq:lower} at $k=2$. For the upper bound we give two routes: an analytic one through a lemma of Zhao \cite[Lemma~4.4]{Zhao2025}, which at the time of writing is a preprint, and a machine one through the value $\sle{7}(\C)=19$, which is part of the archived computation described in Section~\ref{subsec:machine} and was recomputed for this paper (Section~\ref{subsec:thresholds}).

\begin{lemma}[{\cite[Lemma~4.4]{Zhao2025}}]\label{lem:zhao}
Let $p$ be a prime, $G$ a finite abelian $p$-group, and $T$ a zero-sum sequence over $G$ with $|T|\ge2k$, where $2k\ge\Dk(G)+2$. If for some $i\in[1,2k-\Dk(G)]$
\[
a_i=\binom{|T|-k}{k-i}+(-1)^i\binom{|T|-k+i-1}{k-1}\not\equiv0\pmod p,
\]
then $T$ has a nonempty zero-sum subsequence of length at most $k-1$.
\end{lemma}

\begin{proposition}\label{prop:D2}
$\Dk_2(\C)=20$.
\end{proposition}

\begin{proof}
Only $\Dk_2(\C)\le20$ remains. \emph{First route.} Let $S$ be any sequence of length $20$ and put $T=S\cdot g$ with $g=-\sigma(S)$, a zero-sum sequence of length $21$. By Olson \cite{Olson1969}, $\Dk(\C)=13$. Apply Lemma~\ref{lem:zhao} with $p=5$, $k=8$ and $i=2$: the hypotheses $21\ge16$, $16\ge15$ and $2\in[1,3]$ hold, and
\[
a_2=\binom{13}{6}+\binom{14}{7}=1716+3432=5148\equiv3\pmod{5}.
\]
Hence $T$ has a nonempty zero-sum subsequence $U$ of length at most $7$. Its complement $C=T\cdot U^{-1}$ is zero-sum of length at least $14$. Removing one element from $C$ leaves a sequence of length at least $13=\Dk(\C)$, which contains a nonempty zero-sum subsequence $V$; $V$ is a proper subsequence of $C$, so $W=C\cdot V^{-1}$ is zero-sum and nonempty. Thus $T=U\cdot V\cdot W$ with $U,V,W$ pairwise disjoint nonempty zero-sum sequences. At most one of them contains $g$, so $S$ has two disjoint nonempty zero-sum subsequences. \emph{Second route.} By the machine value $\sle{7}(\C)=19$ of Section~\ref{subsec:machine}, Lemma~\ref{lem:recurrence} with $k=1$ and $\ell=7$ gives $\Dk_2(\C)\le\max\{13+7,\,19-1\}=20$.
\end{proof}

\subsection{Machine-verified inputs}\label{subsec:machine}

Three values enter the proof of Theorem~\ref{thm:main}(b) that are not proved analytically in this paper. They were established by exhaustive machine computation whose programs and records accompany the paper as ancillary files, and we state for each exactly what was computed, on what the computation relies, and how it has been checked. Only the upper bounds $\Dk_3(\C)\le25$, $\sle{6}(\C)\le25$ and $\sle{7}(\C)\le19$ are used below; the exact values are recorded for completeness.

\begin{enumerate}[label=\textup{(M\arabic*)},leftmargin=3em]
\item\label{M1} $\Dk_3(\C)=25$. The lower bound is \eqref{eq:lower} at $k=3$; the archive also records the explicit witness $e_1^4\,e_2^4\,e_3^9\,(e_1+e_2)^2\,(e_1+e_3)^2\,(e_2+e_3)^3$ of length $24$ with two but not three pairwise disjoint nonempty zero-sum subsequences, checked by two separately written routines. The upper bound was obtained as follows. Suppose $S$ is a sequence of length $25$ with fewer than three pairwise disjoint nonempty zero-sum subsequences. If $A$ is any nonempty zero-sum subsequence of $S$, then $S\cdot A^{-1}$ has no two disjoint zero-sum subsequences, so $25-|A|\le\Dk_2(\C)-1=19$ by Proposition~\ref{prop:D2}, and $|A|\ge6$; by \ref{M2} below, $S$ has a zero-sum subsequence of length at most six. Hence $S=R\cdot A$ with $|A|=6$, $\sigma(A)=0$, and $R$ a sequence of length $19$ with no two disjoint zero-sum subsequences, necessarily of rank three because $\Dk(C_5^2)=9$ and $19\ge2\cdot9+1$. The archived computation enumerated all such $R$ in normalized form (support containing $e_1,e_2,e_3$ with nondecreasing multiplicities, which every rank-three $R$ can be brought to by $\mathrm{GL}(3,5)$ and a coordinate permutation), $98{,}622$ of them, then all zero-sum extensions $A$ of length six for which $R\cdot A$ has no zero-sum subsequence of length at most five (a necessary condition by the bound $|A|\ge6$ applied to $R\cdot A$), $230{,}983$ candidates in all, and found that every candidate has three pairwise disjoint nonempty zero-sum subsequences. Hence no such $S$ exists and $\Dk_3(\C)\le25$. This computation therefore rests on Proposition~\ref{prop:D2} and on \ref{M2}. Its provenance is a single pair of programs (enumeration and extension, then the three-disjoint test) with a four-input control harness; a second generator later regenerated both candidate sets with the same counts, $98{,}622$ and $230{,}983$, but the final three-disjoint test was not re-executed on them, and an earlier attempt to replay the enumeration on another machine terminated before producing a result. We have not repeated this computation.
\item\label{M2} $\sle{6}(\C)=24$; equivalently, the longest sequence over $\C$ without a nonempty zero-sum subsequence of length at most six has length $23$. The archived value comes from a single exhaustive search program over normalized sequences (rank three) together with a direct enumeration inside $C_5^2$ (rank at most two). For this paper the value was recomputed by an independently written program of a different design, which also reproduces the neighbouring values $\sle{7}(\C)=19$ and $\sle{8}(\C)=18$ recorded in the archive (Section~\ref{subsec:thresholds}).
\item\label{M3} $\sle{7}(\C)=19$, used only in the second route of Proposition~\ref{prop:D2}; same provenance and recomputation as \ref{M2}.
\end{enumerate}

\subsection{Reduction of part (b) to part (a)}

\begin{lemma}\label{lem:corridor}
Assume \ref{M1} and \ref{M2}. Then $\Dk_4(\C)\in\{30,31\}$.
\end{lemma}

\begin{proof}
Lemma~\ref{lem:recurrence} with $k=3$ and $\ell=6$ gives $\Dk_4(\C)\le\max\{25+6,\,24-1\}=31$; the lower bound is \eqref{eq:lower}.
\end{proof}

\begin{lemma}\label{lem:reduction}
Assume \ref{M1}. If $\Dk_4(\C)=31$, then there is a zero-sum sequence over $\C$ of length $31$ with no nonempty zero-sum subsequence of length at most five.
\end{lemma}

\begin{proof}
By Lemma~\ref{lem:Dk-char} there is a zero-sum sequence $S$ of length $31$ with $\zz(S)\le4$. Let $A$ be a nonempty zero-sum subsequence of $S$ of minimum length $m$. Then $S\cdot A^{-1}$ is zero-sum, and $\zz(S\cdot A^{-1})\le3$, since four disjoint zero-sum subsequences of $S\cdot A^{-1}$ together with $A$ would give $\zz(S)\ge5$. By Lemma~\ref{lem:Dk-char} and \ref{M1}, $31-m\le\Dk_3(\C)=25$, so $m\ge6$: every nonempty zero-sum subsequence of $S$ has length at least six.
\end{proof}

\begin{proof}[Proof of Theorem~\ref{thm:main}(b) from part~(a)]
By Lemma~\ref{lem:corridor}, $\Dk_4(\C)\in\{30,31\}$, and by Lemma~\ref{lem:reduction} the value $31$ would produce a sequence excluded by part~(a); hence $\Dk_4(\C)=30$. Together with Proposition~\ref{prop:D2} and \ref{M1} this gives $\Dk_k(\C)=5k+10$ for $k=2,3,4$. For $k\ge4$, Lemma~\ref{lem:recurrence} with $\ell=5$ and the value $\sle{5}(\C)=\eta(\C)=33$ (Section~\ref{subsec:saturation}) gives
\[
\Dk_{k+1}(\C)\le\max\{5k+10+5,\ 32\}=5(k+1)+10,
\]
and \eqref{eq:lower} supplies the reverse inequality. Induction on $k$ completes the proof.
\end{proof}

The rest of the paper proves part~(a), which does not depend on \ref{M1}--\ref{M3} or on Proposition~\ref{prop:D2}.

\section{Structure of a hypothetical counterexample}\label{sec:obstruction}

From now on let $S$ be a zero-sum sequence over $G=\C$ of length $31$ that is short-zero-free; we derive constraints on $S$ that will eventually exclude it. Since $0$ is a zero-sum subsequence of length one and $5g=0$ for every $g$, every element of $S$ is nonzero and has multiplicity at most four. In particular $|\supp(S)|\ge8$, because $4\cdot7=28<31$.

\subsection{Saturation and the multiplicity set $\{1,2,4\}$}\label{subsec:saturation}

We use two results of Gao, Hou, Schmid and Thangadurai \cite{GaoHouSchmidThangadurai2007} on $\C$, also recorded in \cite[Lemma~4.2]{FanGaoWangZhongZhuang2012}. Their Theorem~1.7 gives $s(C_n^3)=\eta(C_n^3)+n-1=9n-8$ for $n=3^a5^b$, hence $\eta(\C)=33$; and their Proposition~5.6 proves that $\C$ has Property~D, which, as they note, implies \emph{Property~C}: every short-zero-free sequence over $\C$ of length $32=\eta(\C)-1$ has the form $T^4$ for a sequence $T$ of eight distinct elements (distinct, because an element repeated in $T$ would have multiplicity eight in $T^4$). We also use $\eta(C_5^2)=13$, an instance of $\eta(C_n^2)=3n-2$ \cite{GaoGeroldinger2006}.

Call a short-zero-free sequence $S'$ \emph{saturated} if $S'\cdot g$ is not short-zero-free for every $g\in G$.

\begin{lemma}\label{lem:saturated}
If $|\supp(S)|\ge9$, then $S$ is saturated.
\end{lemma}

\begin{proof}
If $S\cdot g$ were short-zero-free for some $g$, it would be a short-zero-free sequence of length $32$ with at least nine distinct elements, contradicting Property~C.
\end{proof}

\begin{lemma}[saturation defect]\label{lem:defect}
Let $p$ be an odd prime, $H$ an elementary abelian $p$-group, and $S'$ a saturated sequence over $H$ without nonempty zero-sum subsequences of length at most $p$. If $x\in\supp(S')$ has multiplicity $m<p$, then $S'$ has a subsequence $R$ with $|R|\le p-1-m$, $x\notin\supp(R)$ and $\sigma(R)=-(m+1)x$.
\end{lemma}

\begin{proof}
Since $S'$ is saturated, $S'\cdot x$ has a nonempty zero-sum subsequence $U$ of length at most $p$, and $U$ contains the appended copy of $x$ because $S'$ itself has none. $U$ also contains every original copy of $x$: otherwise, replacing the appended copy in $U$ by an unused original copy would give a zero-sum subsequence of $S'$ of length at most $p$. Deleting the $m+1$ copies of $x$ from $U$ leaves the required $R$.
\end{proof}

\begin{corollary}\label{cor:mult}
If $|\supp(S)|\ge9$, then every element of $S$ has multiplicity $1$, $2$ or $4$.
\end{corollary}

\begin{proof}
Multiplicity $3=p-2$ is impossible: Lemma~\ref{lem:defect} would give $R$ with $|R|\le1$ and $\sigma(R)=-4x=x$, so either $R$ is empty and $x=0$, or $R=y$ with $y=x\in\supp(R)$; both are excluded. Multiplicities $\ge5$ are excluded because $5x=0$.
\end{proof}

\subsection{The multiplicity grammar}

Write $a_1$, $b_2$, $c_4$ for the numbers of elements of $S$ with multiplicity $1$, $2$, $4$, and $s=|\supp(S)|$. When $s\ge9$, Corollary~\ref{cor:mult} gives
\begin{equation}\label{eq:grammar}
a_1+b_2+c_4=s,\qquad a_1+2b_2+4c_4=31,
\end{equation}
hence $b_2=31-s-3c_4$ and $a_1=2s-31+2c_4$; in particular $a_1$ is odd, so $a_1\ge1$. Nonnegativity of $a_1$ and $b_2$ determines, for each $s$, the admissible triples $(a_1,b_2,c_4)$, which we call \emph{patterns}. For $14\le s\le31$ there are exactly $60$ patterns: $42$ with $s\le22$ and $18$ with $s\ge23$ (Table~\ref{tab:counts}). For $s\ge23$ every pattern has $c_4\le2$, because $3c_4\le31-s\le8$.

\subsection{A projective-line restriction}

The nonzero elements of $G$ fall into $31$ \emph{projective lines} $\{x,2x,3x,4x\}$. Restricting $S$ to one line gives a multiplicity vector in $\{0,1,2,4\}^4$ with respect to $(x,2x,3x,4x)$; such a vector is admissible if it carries no nonempty zero-sum sub-multiset of size at most five, i.e.\ no $(t_1,t_2,t_3,t_4)$ with $0<t_1+t_2+t_3+t_4\le5$, $t_j\le$ the $j$-th multiplicity, and $t_1+2t_2+3t_3+4t_4\equiv0\pmod5$. A direct enumeration of the $4^4=256$ vectors shows that exactly $21$ are admissible. The next lemma can be read off from that list, or checked by hand from the four short zero sums $x^3\cdot2x$, $x\cdot x\cdot3x$, $x\cdot4x$ and $x\cdot2x\cdot2x$: the first three exclude a second occupied point on the line of an element of multiplicity four, and the last two together with $x\cdot4x$ exclude two elements of multiplicity at least two on one line.

\begin{lemma}\label{lem:line}
An element of multiplicity four in $S$ is the only element of $S$ on its projective line, and two elements of $S$ of multiplicity at least two are never on a common line. Consequently any two elements of multiplicity at least two are linearly independent.
\end{lemma}

\subsection{Rank}

\begin{lemma}\label{lem:rank}
$\rank(S)=3$. Moreover, if $H$ is a subsequence of $S$ whose support lies in a $2$-dimensional subspace, then $|H|\le12$.
\end{lemma}

\begin{proof}
A subsequence of $S$ supported in a $2$-dimensional subspace is a short-zero-free sequence over a group isomorphic to $C_5^2$, hence has length less than $\eta(C_5^2)=13$. Applied to $S$ itself, $31>12$ forces $\rank(S)=3$.
\end{proof}

\subsection{Exclusion of supports at most $13$}

\begin{proposition}\label{prop:small-support}
$|\supp(S)|\ge14$.
\end{proposition}

This is the one place where part~(a) rests on earlier archived computations rather than on the search of Section~\ref{sec:search}; we describe them so that their scope is clear. Supports at most seven are impossible, as noted above. For support $8$ the pattern is $4^7\,3$; if $x$ is the element of multiplicity three, then $S\cdot x$ is again short-zero-free (a new short zero sum would have to contain all four copies of $x$, hence be $x^4\cdot y$ with $y=-4x=x$, which needs a fifth copy), so by Property~C $\supp(S)=\supp(S\cdot x)$ is one of the extremal $8$-sets $T$ with $T^4$ short-zero-free, and $\sigma(S)=0$ forces $x=-\sigma(T)\in T$. An exhaustive enumeration of these sets in normalized form (support containing $e_1,e_2,e_3$) finds $564$ of them and none with $-\sigma(T)\in T$. For supports $9$ and $10$, Corollary~\ref{cor:mult} leaves the patterns $4^7\,2\,1$, $4^7\,1^3$ and $4^6\,2^3\,1$; four elements of multiplicity four contribute $16>12$ terms, so by Lemma~\ref{lem:rank} three of them are linearly independent and may be normalized to the standard basis, after which an exact search of the kind described in Section~\ref{sec:search} finds no completion. For supports $11$--$13$ the nine admissible patterns are $(1,5,5)$, $(3,2,6)$, $(1,7,4)$, $(3,4,5)$, $(5,1,6)$, $(1,9,3)$, $(3,6,4)$, $(5,3,5)$ and $(7,0,6)$ in the notation $(a_1,b_2,c_4)$; all but $(1,9,3)$ have $c_4\ge4$ and are normalized as above, while $(1,9,3)$ splits according to whether its three elements of multiplicity four span a plane or the whole space, as in Case~(L3) of Section~\ref{sec:cover} (the archived run of the plane case fixes only the three elements of multiplicity four; the re-run described next also normalizes an outside element of multiplicity two, as in (L3)). Two exact engines with different state representations agree that every one of these branches has no completion. For this paper the supports $9$--$13$ searches were re-run by the third implementation of Section~\ref{subsec:engines}, which also found no completion, and the archived support-$8$ program was re-executed with its published expected values reproduced.

\section{The branch cover}\label{sec:cover}

Fix a pattern $(a_1,b_2,c_4)$ with $14\le s\le31$. Let $H$ be the subsequence of $S$ formed by the elements of multiplicity two or four (the \emph{high} part), with $h=b_2+c_4$ support points and length $M_H=2b_2+4c_4$. The group $\mathrm{GL}(3,5)$ acts on sequences over $G$, preserving lengths, sums, multiplicities, ranks and short zero sums; we may therefore replace $S$ by any image under this action. A \emph{branch} is a choice of three linearly independent elements of $\supp(S)$ with prescribed multiplicities, normalized to the standard basis $e_1,e_2,e_3$, together with a constraint on where the remaining elements of multiplicity two may lie. The following two propositions show that every $S$ with $14\le s\le31$ satisfies the hypotheses of at least one of $78$ branches. They only use Lemmas~\ref{lem:line} and~\ref{lem:rank}; in particular, whenever the elements of $H$ have rank three, linear algebra (extension of an independent set to a basis) provides a basis inside $\supp(H)$ whose multiplicity profile is determined by $c_4$, because any two elements of $H$ are independent and elements of multiplicity four are extended first.

\begin{proposition}[supports $14$--$22$: $51$ branches]\label{prop:lower}
Let $14\le s\le22$.
\begin{enumerate}[label=\textup{(L\arabic*)},leftmargin=3em]
\item If $c_4\le2$, then $M_H=31-a_1\ge14$, so by Lemma~\ref{lem:rank} $\rank(H)=3$, and $\supp(H)$ contains a basis with multiplicity profile $(2,2,2)$, $(4,2,2)$ or $(4,4,2)$ according as $c_4=0,1,2$. Normalize it to $(e_1,e_2,e_3)$. (One branch per pattern.)
\item If $c_4\ge4$, then the elements of multiplicity four alone contribute $16>12$ terms, so three of them are independent; normalize them to $(e_1,e_2,e_3)$. (One branch per pattern.)
\item If $c_4=3$, the three elements of multiplicity four contribute $12$ terms, and their rank is either three or two. In the first case normalize them to $(e_1,e_2,e_3)$. In the second case they span a plane $P$; normalize two of them to $e_1,e_2$, so that the third is $\alpha e_1+\beta e_2$ with $\alpha,\beta\ne0$ by Lemma~\ref{lem:line}. No element of multiplicity two lies in $P$, since it would raise the number of terms of $S$ supported in $P$ to $14>12$; so if $b_2>0$ some element of multiplicity two lies outside $P$, and if $b_2=0$ some element of multiplicity one lies outside $P$ by Lemma~\ref{lem:rank}. Normalize that outside element to $e_3$ (the pointwise stabilizer of $e_1,e_2$ acts transitively on the vectors outside $P$), with multiplicity $2$ or $1$ respectively. (Two branches per pattern.)
\end{enumerate}
Among the $42$ patterns, $9$ have $c_4=3$, giving $42+9=51$ branches.
\end{proposition}

\begin{proposition}[supports $23$--$31$: $27$ branches]\label{prop:upper}
Let $23\le s\le31$, so that $c_4\le2$.
\begin{enumerate}[label=\textup{(U\arabic*)},leftmargin=3em]
\item If $h\ge3$ and $\rank(H)=3$, then $\supp(H)$ contains a basis with multiplicity profile $(2,2,2)$, $(4,2,2)$ or $(4,4,2)$ according as $c_4=0,1,2$; normalize it to $(e_1,e_2,e_3)$.
\item If $h\ge2$ and $\rank(H)=2$, then $M_H\le12$ by Lemma~\ref{lem:rank}, two elements of $H$ (of multiplicity four first, then two) are normalized to $e_1,e_2$, every remaining element of multiplicity two lies in the plane $\langle e_1,e_2\rangle$, and, since $\rank(S)=3$, some element of multiplicity one lies outside that plane and is normalized to $e_3$. This case is void when $h\ge3$ and $M_H\ge13$ (equivalently $M_H\ge14$, as $M_H$ is even); when $h=2$ it is the only case, since $\rank(H)\le2$.
\item If $h=1$, the unique element of $H$ and two elements of multiplicity one form a basis, by Lemma~\ref{lem:rank}; normalize them to $(e_1,e_2,e_3)$.
\item If $h=0$, three elements of multiplicity one form a basis; normalize them to $(e_1,e_2,e_3)$.
\end{enumerate}
Applying (U1)--(U4) to the $18$ patterns gives $27$ branches: $12$ patterns with $h\ge3$ each give a rank-three branch, of which $9$ have $M_H\le12$ and also give a rank-two branch; $3$ patterns have $h=2$, $2$ have $h=1$ and $1$ has $h=0$.
\end{proposition}

The per-support counts are listed in Table~\ref{tab:counts}. None of the counts $51$, $27$, $78$ is an input to any program: the branch list is generated from \eqref{eq:grammar} and the case distinctions above, and it was regenerated by separately written programs (Section~\ref{subsec:verification}).

\begin{table}[htbp]
\centering
\caption{Admissible multiplicity patterns and branches by support size $s$. The total of $60$ patterns and $78$ branches follows from \eqref{eq:grammar} and Propositions~\ref{prop:lower}--\ref{prop:upper}.}
\label{tab:counts}
\small
\begin{tabular}{@{}rrr@{\qquad}rrr@{}}
\toprule
$s$ & patterns & branches & $s$ & patterns & branches\\
\midrule
14 & 4 & 5 & 23 & 3 & 4\\
15 & 5 & 6 & 24 & 3 & 5\\
16 & 6 & 7 & 25 & 3 & 5\\
17 & 5 & 6 & 26 & 2 & 4\\
18 & 5 & 6 & 27 & 2 & 3\\
19 & 5 & 6 & 28 & 2 & 3\\
20 & 4 & 5 & 29 & 1 & 1\\
21 & 4 & 5 & 30 & 1 & 1\\
22 & 4 & 5 & 31 & 1 & 1\\
\midrule
$14$--$22$ & 42 & 51 & $23$--$31$ & 18 & 27\\
\bottomrule
\end{tabular}
\end{table}

\section{The exhaustive search}\label{sec:search}

\subsection{What is enumerated}\label{subsec:enumeration}

Each branch fixes a pattern, a set of \emph{seeds} (the normalized elements with their multiplicities) and possibly a plane constraint, and asks whether the seeds can be completed to a zero-sum short-zero-free sequence of that pattern. A search program enumerates the remaining elements of multiplicity four, then of multiplicity two, then of multiplicity one, each class in increasing order of the encoding fixed in Section~\ref{sec:ingredients}, so that each multiset of support points is visited once; in the rank-two branches of (U2) the remaining elements of multiplicity two are confined to the plane $\langle e_1,e_2\rangle$. (In the $69$ branches whose seeds already contain every element of multiplicity four, which is the case whenever $c_4\le2$, the programs have no multiplicity-four stage.) Elements equal to a seed are skipped because support points are distinct, and elements on the projective line of an element of multiplicity four are skipped by Lemma~\ref{lem:line}; no other symmetry reduction is applied, so the enumeration may visit several $\mathrm{GL}(3,5)$-equivalent completions, which affects running time but not completeness. When all but one element of multiplicity one have been chosen, the last one is determined by $\sigma(S)=0$; it is accepted only if it is nonzero, distinct from the chosen elements, admissible for Lemma~\ref{lem:line}, larger in the encoding than the last chosen element of multiplicity one, and does not create a short zero sum. The order condition ensures that every admissible multiset of free elements of multiplicity one is met exactly once, as the ordered list of its smaller elements followed by its largest one.

Short zero sums are detected incrementally. The state of a partial sequence records, for each weight $w\in\{0,\dots,5\}$, the set of group elements that are sums of exactly $w$ of its terms (counted with multiplicity). Adding one copy of an element $g$ replaces the weight-$w$ set by its union with the translate by $g$ of the weight-$(w-1)$ set, for $w=5,\dots,1$, and the partial sequence is rejected as soon as $0$ enters a set of positive weight. Since the state before the addition already contains no zero sum of positive weight, this test is exact: a partial sequence survives if and only if it is short-zero-free. Hence a branch reports no completion if and only if no sequence satisfying its hypotheses exists. We call a partial sequence accepted by the test a \emph{node}, and a node in which all free elements but the forced one have been chosen a \emph{leaf}.

\subsection{Three implementations and a cross-machine replay}\label{subsec:engines}

The archived computation attacks each branch with two engines built on different state representations: the first keeps one $125$-bit mask per weight and translates it by coordinate-wise shifts, while the second packs five $25$-bit coordinate planes into vector lanes with separately written cyclic shifts. Three branch-family drivers times two representations give six programs. Across the $78$ branches the $156$ runs complete, agree on the number of nodes and leaves in every branch, and report zero completions in every branch. Every branch has zero leaves: no partial sequence ever survives to the forced last element, so the theorem is decided entirely by the incremental test at the multiplicity-two and multiplicity-one stages. The largest branch is the all-singleton pattern with $s=31$, with $1{,}009{,}511{,}446$ nodes; the total over all branches is $2{,}943{,}691{,}753$ nodes. A replay on a second machine, from the same archived sources but with a different compiler major version and a different native instruction target, reproduced the results and digests exactly; this is an independent execution of the same programs, not an independent implementation.

For the present paper a third engine was written without reusing code from the two archived ones. It necessarily follows the same enumeration protocol (order of stages, skip rules, forced last element, node counting), since that protocol is part of the case analysis; its independence lies in the state representation and the zero-sum test. It stores each weight set as a $125$-bit set in two machine words and translates it by table-driven byte permutation, checking every permutation-table entry at start-up against a direct bit-by-bit translation, and it detects a new short zero sum by testing whether $-g$ already lies in the weight-$(w-1)$ set rather than by translating that set. It was validated on controls with known answers before being pointed at the $78$ branches: it finds a short-zero-free sequence of the form $T^4$ of length $32$ (which exists since $\eta(\C)=33$) and, with the forbidden zero-sum length lowered from five to one, two or three, finds zero-sum completions of test patterns, each reported sequence being re-checked by a separate brute-force routine; and it reports no completion for the supports $9$--$13$ patterns of Proposition~\ref{prop:small-support}. Run on the $78$ branches (the nine plane branches of (L3) as sixteen sub-runs each, one per in-plane third element, $213$ runs in all), it reports no completion in any branch, and its counts of nodes and leaves agree with the archived engines branch for branch, the sixteen sub-run counts summing to the archived count. All programs, result files and receipts described in this section, together with the source of this paper, accompany it as ancillary files, so that every checker can be re-run without modification.

\subsection{Second implementation of the threshold values}\label{subsec:thresholds}

A separate program of the same design decides, for $T\in\{6,7,8\}$ and a target length $L$, whether a sequence over $\C$ of length $L$ without a nonempty zero-sum subsequence of length at most $T$ exists, separately for rank three (support normalized to contain $e_1,e_2,e_3$ with nondecreasing multiplicities), rank two (support in $\langle e_1,e_2\rangle$ containing $e_1,e_2$) and rank one, with every reported sequence re-checked by brute force. For $T=6$ it finds a sequence of length $23$ and none of length $24$ in rank three, a longest sequence of length $11$ in rank two and of length $4$ in rank one; hence $\sle{6}(\C)=24$, which is \ref{M2}. For $T=7$ the corresponding maxima are $18$, $10$ and $4$, so $\sle{7}(\C)=19$, which is \ref{M3}; for $T=8$ the rank-three maximum is $17$, so $\sle{8}(\C)=18$. All three values agree with the archived spectrum, which was obtained by a differently designed program.

\subsection{Independent reconstruction of the cover and adversarial controls}\label{subsec:verification}

Completeness of the case analysis is as important as the searches. Two programs, written separately from the search engines, regenerate the cover from the mathematics alone, that is from \eqref{eq:grammar}, the primitive enumeration of admissible line states, the value $\eta(C_5^2)=13$ and the case distinctions of Propositions~\ref{prop:lower} and~\ref{prop:upper}; neither reads the archived branch list, engine sources or outputs before generating its own list. The first, from the archive, regenerates the $60$ patterns, the $21$ admissible line states and the $78$ branches identified by pattern and case, and compares them with the archived cover: no branch missing and none extra. It also rejects hostile perturbations of the analysis: omitting one branch ($77$), changing the total length in \eqref{eq:grammar} to $30$ ($56$ patterns), dropping the rank split in (L3) ($69$ branches), or weakening the threshold $\eta(C_5^2)=13$ to $15$ ($80$ branches). The second, the driver of the third engine, regenerates the $60$ patterns and the $78$ branches together with their seeds and plane constraints and finds them identical to the archived cover before any branch is run. On the result side, the archived checker that validates the result file rejects a copy in which one branch is deleted, one engine source is altered, one output digest is corrupted, the zero-completion flag is negated, or a completion count is set to one, while accepting the unmodified file; the two checkers that validate only the cover accept those copies, as they should, since the mutations do not touch the cover.

\section{Proof of Theorem~\ref{thm:main}(a)}\label{sec:proof}

Suppose $S$ is a zero-sum sequence over $\C$ of length $31$ without a nonempty zero-sum subsequence of length at most five. By Proposition~\ref{prop:small-support}, $|\supp(S)|\ge14$, so $S$ is saturated (Lemma~\ref{lem:saturated}) and every multiplicity lies in $\{1,2,4\}$ (Corollary~\ref{cor:mult}). Hence $(a_1,b_2,c_4)$ is one of the $60$ patterns of \eqref{eq:grammar} with $14\le s\le31$. By Propositions~\ref{prop:lower} and~\ref{prop:upper}, an image of $S$ under $\mathrm{GL}(3,5)$ satisfies the hypotheses of one of the $78$ branches, and the exhaustive search of Section~\ref{sec:search} shows that no zero-sum short-zero-free sequence satisfies the hypotheses of any branch. This contradiction proves part~(a); part~(b) was derived from part~(a) in Section~\ref{sec:ingredients}. \qed

\section{Relation to prior work and scope of the claim}\label{sec:prior}

The lower bound \eqref{eq:lower} is Freeze and Schmid's; the recurrence of Lemma~\ref{lem:recurrence} and the eventual arithmetic-progression behaviour of $k\mapsto\Dk_k(G)$ are theirs as well \cite{FreezeSchmid2010}. In their notation the theorem says that $\Dk_k(\C)=\Dk_0(\C)+5k$ with $\Dk_0(\C)=10$ from $k=2$ onward. The value $\Dk_2(\C)=20$ is an immediate consequence of that bound with either Zhao's lemma and Olson's theorem or the machine value $\sle{7}(\C)=19$, and we do not regard it as new. Saturation arguments, the projective-line restriction and normalization by $\mathrm{GL}(3,5)$ are standard devices of the area. What this paper adds is the exact value $\Dk_4(\C)=30$ with its consequence $\Dk_k(\C)=5k+10$ for all $k\ge2$, the finite reduction of Sections~\ref{sec:obstruction}--\ref{sec:cover} that makes the remaining case decidable, and the verified computation that decides it. The route (saturation defect, multiplicity grammar, rank/plane branches normalized to a basis, exhaustive weight-layered search) is not specific to the numbers involved beyond $p=5$ and the length $31$, but Lemma~\ref{lem:saturated} needs Property~C or another structure theorem at length $\eta(G)-1$ for the group in question, which restricts its immediate reuse to groups where such a theorem is available.

\begin{remark}\label{rem:C23}
For $C_2^3$, \eqref{eq:FS} with $s=3$, $t=1$ and $\delta=0$ gives $\Dk_k(C_2^3)\ge4+3+2(k-2)=2k+3$ for $k\ge2$. Conversely $\Dk(C_2^3)=4$ and $\sle{2}(C_2^3)=8$ (eight elements of a group of order eight include $0$ or a repeated element), so Lemma~\ref{lem:recurrence} with $\ell=2$ gives $\Dk_2(C_2^3)\le\max\{4+2,\,7\}=7$ and $\Dk_{k+1}(C_2^3)\le\max\{\Dk_k(C_2^3)+2,\,7\}$, whence $\Dk_k(C_2^3)=2k+3$ for all $k\ge2$ by induction. Thus the Freeze--Schmid bound is attained from $k=2$ onward for $C_2^3$ and $\C$ but not for $C_3^3$. For which primes $p$ the bound \eqref{eq:FS}, namely $\Dk_k(C_p^3)\ge3p-2+3\lfloor p/2\rfloor+\delta+(k-2)p$, is attained for all $k\ge2$, and where the threshold in $k$ sits otherwise, we do not know.
\end{remark}

To the best of our knowledge, no exact value of $\Dk_k(\C)$ for any $k\ge2$ has appeared in the literature. Our search covered the arXiv record and public indexes with the vocabulary of generalized, $k$-th and multiwise Davenport constants and of short zero-sum obstructions; the closest results are the rank-two formula and its inverse theory \cite{Zhong2025}, the elementary $2$-group values \cite{DelormeOrdazQuiroz2001,FreezeSchmid2010}, the $C_3^3$ values \cite{BhowmikSchlagePuchta2007}, and work on other zero-sum invariants of rank-three groups over $C_3$ \cite{Zhang2023}. The route from $\Dk(G)$ to $\Dk_k(G)$ for $p$-groups with $\Dk(G)\le2\exp(G)-1$ (see \cite[Section~6.1]{GeroldingerHalterKoch2006}) does not apply to $\C$, where $13>9$.

\section{Provenance of the computations and limits of the claim}\label{sec:deferred}

Theorem~\ref{thm:main}(a) is a computer-assisted theorem whose non-computational part is written out above. Its computational part consists of the $78$-branch search, executed by three independently written implementations (two archived engines and the engine written for this paper), replayed on a second machine, and tied to a case analysis regenerated by two separately written programs, together with the archived exclusions of supports $8$--$13$, of which supports $9$--$13$ were re-run here. Part~(b) additionally rests on Proposition~\ref{prop:D2} (Zhao's preprint lemma, or the machine value $\sle{7}(\C)=19$), on the machine-verified inputs \ref{M1}--\ref{M2}, of which \ref{M1} has single-implementation provenance while \ref{M2} was recomputed here, and on the cited values $\eta(\C)=33$ with Property~C and $\eta(C_5^2)=13$.

The archived programs and records originate in the author's research programme, and the programs written for this paper (the third engine, its driver, and the threshold program) were developed without access to the archived engine code beyond the branch specifications and the enumeration protocol. Generative AI tools were used in preparing the programs and the text; the author is responsible for all content. No person outside that process has reviewed the case analysis or the programs; no externally produced or externally checked certificates for the individual branches exist; and the paper has not been refereed. These are deferred to the refereeing of this paper or to later independent replication, and nothing in the paper should be read as asserting that they have occurred.

\bibliographystyle{plainnat}
\bibliography{bibliography}
\end{document}